\documentclass[9pt,twocolumn]{IEEEtran}
\usepackage{amsmath,amssymb,euscript,yfonts,psfrag,latexsym,dsfont,graphicx,bbm,color,amstext,wasysym,subfig,parskip,soul,bm}
\usepackage{amsthm}
\usepackage[dvipsnames]{xcolor}

\graphicspath{{./},{./figures/}}

\begin{document}
\newtheorem{thm}{Theorem}
\newtheorem{cor}[thm]{Corollary}
\newtheorem{conj}[thm]{Conjecture}
\newtheorem{lemma}[thm]{Lemma}
\newtheorem{prop}[thm]{Proposition}
\newtheorem{problem}{Problem}
\newtheorem{remark}{Remark}
\newtheorem{definition}{Definition}
\newtheorem{example}{Example}

\newcommand{\bp}{{\bm{p}}}
\newcommand{\bq}{{\bm{q}}}
\newcommand{\bc}{{\bm{c}}}
\newcommand{\be}{{\bm{e}}}
\newcommand{\br}{{\bm{r}}}
\newcommand{\bv}{{\bm{v}}}
\newcommand{\bx}{{\bm{x}}}
\newcommand{\bbf}{{\bm{f}}}
\newcommand{\balpha}{{\bm{\alpha}}}
\newcommand{\bbeta}{{\bm{\beta}}}
\newcommand{\bdelta}{{\bm{\delta}}}
\newcommand{\bPi}{{\bm{\Pi}}}
\newcommand{\bC}{{\bm{C}}}
\newcommand{\bL}{{\bm{L}}}
\newcommand{\bI}{{\bm{I}}}
\newcommand{\bP}{{\bm{P}}}
\newcommand{\bD}{{\bm{D}}}
\newcommand{\bQ}{{\bm{Q}}}
\renewcommand{\dagger}{{T}}

\newcommand{\differential}{{\rm{d}}}

\newcommand{\interior}[1]{%
  {\kern0pt#1}^{\mathrm{o}}%
}

\newcommand{\mR}{{\mathbb R}}
\newcommand{\mZ}{{\mathbb Z}}
\newcommand{\Prob}{{\rm Pr}}
\newcommand{\D}{{\mathbb D}}
\newcommand{\cX}{{\mathcal X}}
\newcommand{\cT}{{\mathcal T}}
\newcommand{\cS}{{\mathcal S}}
\newcommand{\E}{{\mathbb E}}
\newcommand{\mcN}{{\mathcal N}}
\newcommand{\mcR}{{\mathcal R}}
\newcommand{\diag}{\operatorname{diag}}
\newcommand{\tr}{\operatorname{trace}}
\renewcommand{\odot}{\circ}
\newcommand{\ignore}[1]{}

\newcommand{\magenta}{\color{magenta}}
\newcommand{\red}{\color{red}}
\newcommand{\blue}{\color{blue}}
\newcommand{\gray}{\color{gray}}
\definecolor{grey}{rgb}{0.6,0.3,0.3}
\definecolor{lgrey}{rgb}{0.9,.7,0.7}

\newcommand{\symsum}{\displaystyle\sum_{\rm{symm}}}

\def\spacingset#1{\def\baselinestretch{#1}\small\normalsize}
\setlength{\parindent}{20pt}
\setlength{\parskip}{12pt}
\spacingset{1}

\title{\huge Stability Theory 
of Stochastic Models in Opinion Dynamics}

\author{Zahra Askarzadeh, Rui Fu, Abhishek Halder, Yongxin Chen, and Tryphon T.\ Georgiou
\thanks{Z.\ Askarzadeh, R.\ Fu, and T.\ T. Georgiou are with the Department of Mechanical and Aerospace Engineering, University of California, Irvine, CA; emails: zaskarza@uci.edu, rfu2@uci.edu, tryphon@uci.edu}
\thanks{A.\ Halder is with the Department of Applied Mathematics and Statistics, University of California, Santa Cruz,
        CA 95064, USA;
        {ahalder@ucsc.edu}}
\thanks{Y.\ Chen is with the Department of Electrical and Computer Engineering, Iowa State University, IA 50011,
Ames, Iowa, USA; {yongchen@iastate.edu}}
}

\markboth{}{}

\maketitle

\begin{abstract}
We consider a certain class of nonlinear maps that preserve the probability simplex, i.e., stochastic maps, that are inspired by the DeGroot-Friedkin model of belief/opinion propagation over influence networks. The corresponding dynamical models describe the evolution of the probability distribution of interacting species. Such models where the probability transition mechanism depends nonlinearly on the current state are often referred to as {\em nonlinear Markov chains}. In this paper we develop stability results and study the behavior of representative opinion models. The stability certificates are based on the contractivity of the nonlinear evolution in the $\ell_1$-metric.  We apply the theory to two types of opinion models where the adaptation of the transition probabilities to the current state is exponential and linear, respectively--both of these can display a wide range of behaviors. We discuss continuous-time and other generalizations.
\end{abstract}

\noindent{\bf Keywords:} $\ell_1$-stability of stochastic maps, nonlinear  Markov semigroups, opinion dynamics, influence networks, reflected appraisal.

\section{Introduction}
Models of social interactions and the formation of opinions in large groups have been receiving increasing attention in recent years  (see \cite{jia2015opinion,chen2017distributed,altafini2013consensus,friedkin2015problem,parsegov2016novel,acemouglu2013opinion} and the references therein). As the basis for social exchanges, an averaging mechanism has been postulated in the literature, whereby the outcome represents a weighted sum of individual preferences or beliefs. In turn, the averaging mechanism itself is modified by the outcome of past interactions, reflecting relative increase or decrease in the confidence and, thereby, influence of particular individuals. Such feedback models can be traced to \cite{degroot1974reaching,friedkin2010attitude,FriedkinJohnsen,hegselmann2002opinion}.

Averaging schemes leading to consensus are broadly relevant in coordination of dynamical systems such as co-operating drones or ground robots, sensor networks, formation flight, and distributed frequency regulation in power grid, see, e.g.\ \cite{bullo2009distributed, cao2013overview, zhang2012convergence}. The distinguishing feature of social interaction models has been the postulate of a suitable nonlinear effect that enhances or, perhaps, diminishes the influence of particular individuals in the group. The purpose of this paper is first to step back, and view the dynamics as a nonlinear random walk. We then develop a stability theory for corresponding stochastic maps by resorting to the $\ell_1$ metric. The key element of our approach is to consider the differential of the stochastic maps and assess whether these are contractive in $\ell_1$.

More specifically, we consider a discrete-time (or rather, discrete-indexed, where the index may represent issue being considered) process $\{X_t\mid t\in \mZ_+\}$ taking values on a finite state-space $\cX=\{1,\,2,\ldots,\,n\}$. We denote by $\bp(t)$ the marginal probability vector, i.e., its entry $p_{i}(t)=\Prob(X_t=i)$ is the occupation probability of state $i$ at iteration index $t$, and postulate a transition mechanism that depends nonlinearly on the occupation probability (a.k.a. belief state) of the process according to the rule:
\begin{subequations}\label{eq:MC}
\begin{align}\label{eq:transition}
&\Pi_{ij} := \Prob(X_{t+1}=j \mid X_t=i)= \rho_i(t)\delta_{ij} + (1-\rho_i(t)) C_{ij},\\
&\rho_i(t) :=r(p_{i}(t)),
\end{align}
\end{subequations}
where $\bC:=[C_{ij}]_{i,j=1}^n$ is a row-stochastic matrix\footnote{A matrix $\bC$ is referred to as stochastic (or, row-stochastic, for specificity) provided $C_{ij}\geq 0$ for all $(i,j)$ and $\sum_jC_{ij}=1$. Such matrices map the probability simplex into itself (or into another, if not square).}, $\delta_{ij}$ equals one for $i=j$ and zero otherwise, and $r(\cdot)$ is a differentiable function
\[
r\;:\;[0,1] \mapsto [0,1].
\]
In general, the mapping $r(\cdot)$ needs to be neither onto nor invertible (nor independent of $i$, as taken at the early part of the paper, for simplicity). Typical examples include
\begin{align}\label{eq:options}
r(x)=x,\; 1-x,\;1-e^{-\gamma x},\;{\rm or}\;e^{-\gamma x}, \mbox{ for some } \gamma >0.
\end{align}

Equation \eqref{eq:transition} represents a model for a ``lazy'' random walk where the transition probabilities $C_{ij}$ are modified to increase/decrease the ``prior'' return-probability from $C_{ii}$ to $\rho_i + C_{ii}(1-\rho_i)$, in a way that depends on the probability of the corresponding state, since $\rho_i(t)=r(p_i(t))$. For this reason, we refer to $r(\cdot)$ as the \emph{reinforcement function}. Thus, the essence of the above model is that the random walk adapts the return probability of each state so as to promote or discourage residence in states with high marginal probability. 
An alternative interpretation of the time $t$-marginal probabilities is as representing confidence or influence which, accordingly, is modified constructively or destructively by the likelihood of the particular state of the process. It has been argued, for instance, that high confidence and success in an argument, begets higher confidence.

The model in \eqref{eq:MC} provides an example of a discrete-time, discrete-space {\em nonlinear Markov semigroup} that maps the probability simplex on $\cX$ into itself \cite{kolokoltsov2010nonlinear}. In general, for a nonlinear Markov semigroup to define finite-dimensional distributions (and thereby a random process), one needs to decide on a stochastic representation as in \eqref{eq:MC}, which may not be unique. Then, once the transition probabilities are specified as a nonlinear function of the state, the stochastic process can be defined in the form of a time inhomogeneous-Markov chain \cite[Chapter 1]{kolokoltsov2010nonlinear}.
Such {\em nonlinear Markov} models arise naturally as limits of interacting particle systems that model processes with mass-preserving interactions \cite[Section 1.3]{kolokoltsov2010nonlinear}. Herein, we will not be concerned with the probabilistic nature and properties of such systems, but instead focus on the dynamical response and stability of equilibria on the probability simplex. Thus, for the most part, we will focus on stochastic maps with transition probabilities as above.

In the context of opinion dynamics, the matrix $\bC=[C_{ij}]_{i,j=1}^n$ in \eqref{eq:MC}, encodes the influence of neighboring nodes--a standing assumption throughout is that $\bC$ corresponds to a {\em strongly connected aperiodic Markov chain}. With regard to the reinforcement mechanism, of particular interest are exponentially-scaled transition kernels (introduced here)
\begin{subequations}\label{eq:Gs}
\begin{align}
\Pi_{ij}(x)&= (1-e^{-\gamma x_i})\delta_{ij}+e^{-\gamma x_i}C_{ij}, \label{eq:Gkernel}\mbox{ and its ``opposite''}\\
\overline{\Pi}_{ij}(x)&= e^{-\gamma x_i}\delta_{ij}+(1-e^{-\gamma x_i})C_{ij}, \label{eq:Gbarkernel}
\end{align}
\end{subequations}
as well as the linearly-scaled kernels
\begin{subequations}\label{eq:GsDGF}
\begin{align}
\Pi_{ij}(x)&= \gamma x_i\delta_{ij}+(1-\gamma x_i)C_{ij}, \label{eq:GkernelDGF} \mbox{ and}\\
\overline{\Pi}_{ij}(x)&= (1-\gamma x_i)\delta_{ij}+\gamma x_iC_{ij} \label{eq:GbarkernelDGF}
\end{align}
\end{subequations}
which have been considered in, e.g., \cite{jia2015opinion}.
Naturally, in all these cases, $\bPi=\left[\Pi_{ij}\right]_{i,j=1}^n$ and $\overline{\bPi}=\left[\overline{\Pi}_{ij}\right]_{i,j=1}^n$ are row-stochastic (i.e., rows sum to one).
Those two models will be analyzed in some detail as they provide rather insightful examples of the dynamics that one can expect of such models. We highlight ranges of parameters where globally stable behavior is observed and where the process tends towards a stationary distribution and, others, where multiple equilibria, periodic orbits, or chaotic behavior is observed. Local stability of equilibria (i.e., local stationarity of distributions), if that is the case, can be assessed using the theory developed in Section \ref{sec:stability}.

The evolution of the marginal probability (column) vector  $\bp(t)$ corresponding to \eqref{eq:MC} (also, (\ref{eq:Gs}) and (\ref{eq:GsDGF})) is as follows:
\begin{subequations}\label{eq:lazywalkboth}
\begin{align}\label{eq:lazywalk}
\bp(t+1) &= \bPi(\bp(t))^{\dagger}\bp(t),
\end{align}
with
\begin{align}\label{bPi1}
\bPi(\bp(t))^\dagger&= \bD(\bp(t))+\bC^\dagger (\bI-\bD(\bp(t))),
\end{align}
and a diagonal matrix
\begin{align}\label{bPi2}
\bD(\bp(t))&= \diag(\br(\bp(t))),
\end{align}
\end{subequations}
where ``$(\cdot)^\dagger$'' as usual denotes transposition.
As noted, throughout, $\bC$ is row stochastic and corresponds to a strongly connected and aperiodic chain. The starting point for the evolution is $\bp_0\in \cS_{n-1}$, where
\[\cS_{n-1}:=\{\bx\in \mR^n\mid x_i\geq 0,\;\sum_ix_i=1\}
\]
denotes the probability simplex. By $\cS_{n-1}^o$ we will denote the (open) interior of $\cS_{n-1}$.

A closely related alternative model for the evolution of influence and opinion dynamics that has appeared in the literature, is to postulate the transition mechanism
\begin{eqnarray}\label{eq:DeGrootFriedkin}
&&\bp(t+1) = \left[\bPi(\bp(t))^{\dagger}\right]_{\rm FP},
\end{eqnarray}
where the notation
$\left[ \bPi^\dagger \right]_{\rm FP}$
represents the mapping $\bPi^\dagger \mapsto \bq\in \cS_{n-1}$ of an irreducible (row) stochastic matrix $\bPi$ to its corresponding Frobenius-Perron eigenvector, i.e., to the unique probability (column) vector $\bq$ that satisfies $\bPi^\dagger \bq= \bq$. The relation between the two update-mechanisms, \eqref{eq:lazywalk} and \eqref{eq:DeGrootFriedkin}, can be understood by virtue of the fact that
$
\left(\bPi(\bp(t))^{\dagger}\right)^{k} \bp(t)
$
is approximately equal to the right Frobenius-Perron eigenvector of $\bPi(\bp(t))^{\dagger}$ for sufficiently large $k$, and hence a suitable modification of the dynamics in \eqref{eq:lazywalk} (i.e., by introducing a suitably high exponent) approximates the dynamics in \eqref{eq:DeGrootFriedkin}. We will not be concerned with the update mechanism in \eqref{eq:DeGrootFriedkin}, as our primary interest is in the general transition mechanism \eqref{eq:lazywalk}. It is reasonable to expect that stochastic maps in either form, \eqref{eq:lazywalk}, or \eqref{eq:DeGrootFriedkin}, for specific choices of kernel $\Pi_{ij}(\cdot)$ and generalizations (see Sections \ref{sec:continuous}-\ref{sec:grouping}), have appealing properties as models of opinion dynamics.

The exposition in our manuscript proceeds as follows. In Section \ref{sec:stability}, we provide conditions that ensure contractivity in $\ell_1$ (Theorem \ref{thm1} and Propositions \ref{cor3}, \ref{cor4}), quantify the $\ell_1$-gain (Theorem \ref{thm6}), and give conditions for attractiveness of a periodic orbit (Proposition \ref{cor5}). We discuss ``exponential influence'' models in Section \ref{sec:Gs} and DeGroot-Friedkin models in Section \ref{sec:DGF}. In both sections we present and analyze representative dynamical behaviors via examples. We comment briefly on the continuous-time counterpart of such models and, in Section \ref{sec:grouping} we introduce local coupling in the reinforcement mechanism to model grouping between colluding subgroups in opinion forming, and comment on extensions of the theory to account for such interactions. Section \ref{sec:conclusion} provides concluding remarks and directions.


\section{$\ell_1$-contractivity}\label{sec:stability}
We consider stochastic maps of the particular form
\begin{subequations}\label{eq:class}
\begin{align}\label{eq:nonlinearevolution}
\bbf &:\,\cS_{n-1}\to \cS_{n-1}\,:\,\bp\mapsto \bbf(\bp):=\bPi(\bp)^\dagger \bp = \bq,
\end{align}
where $\bPi(p)$ is of the form
\begin{align}\label{bP3}
\bPi(\bp)^\dagger&=\bC_0^\dagger\bD(\bp) + \bC_1^\dagger (\bI-\bD(\bp)),
\end{align}
\end{subequations}
with $\bC_0$, $\bC_1$ both row stochastic, and $\bD(\bp)$ diagonal with entries bounded by one; the expression \eqref{bPi1} is the special case where $\bC_0$ is the identity matrix.
Note that $\bPi(\bp)$ has nonnegative entries with rows summing to one for all $\bp\in\cS_{n-1}$.
Under suitable conditions, which often hold for the type of dynamics that we consider, $\bbf$ turns out to be contractive, and even strictly contractive\footnote{The map $\bbf$ is strictly contractive on $S\subset\cS_{n-1}$ if there exists $\epsilon>0$ that may depend on $S$ so that
\[
(1-\epsilon)\|\bp^b-\bp^a\|_1 \geq  \|\bbf(\bp^b)-\bbf(\bp^a)\|_1,
\]
for all  $\bp^{a},\bp^{b}\in S$. It is contractive if the statement holds for $\epsilon=0$.} in $\ell_1$, in $\cS_{n-1}$ or subsets thereof as specified.

Denote by $\cT$ the tangent space of the probability simplex, i.e.,
\[
\cT:=\{\bdelta \in \mR^n \mid {\mathds 1}^\dagger \bdelta=0\}
\]
with $\mathds 1$ the column vector of ones.
The Jacobian of $\bbf(\cdot)$ is
\begin{align*}
\differential\bbf\,:\, \cT\to \cT\;:\; (\delta_j)_{j=1}^n \mapsto &\left( \sum_{i=1}^n \Pi_{ij} \delta_i\right)_{j=1}^n\\
&+\left( \sum_{i,k=1}^n \frac{\partial\Pi_{ij}}{\partial p_k} p_i\delta_k\right)_{j=1}^n
\end{align*}
where, by interchanging indices, the latter term can be written as
\[
\left( \sum_{k,i=1}^n \frac{\partial\Pi_{kj}}{\partial p_i} p_k\delta_i\right)_{j=1}^n.
\]
Thus, $\differential\bbf$ can be written in a matrix form as
\begin{equation}\label{eq:oneQ}
\differential\bbf\,:\, \bdelta \mapsto \left(\underbrace{\bPi^\dagger + \left[ \frac{\partial \bPi^\dagger}{\partial p_1}\bp,\ldots, \frac{\partial \bPi^\dagger}{\partial p_n}\bp \right]}_{\bQ^\dagger}\right)\bdelta.
\end{equation}
Since ${\mathds 1}^\dagger\bC_i^\dagger={\mathds 1}^\dagger$, for $i\in\{0,1\}$, the columns on the second entry in the expression for $\bQ^\dagger$ satisfy
\begin{align*}
{\mathds 1}^\dagger\left( \frac{\partial \bPi^\dagger}{\partial p_j}\bp\right) &= {\mathds 1}^\dagger\left(\bC_0^\dagger \frac{\partial \bD}{\partial p_j}\bp -
 \bC_1^\dagger\frac{\partial \bD}{\partial p_j}\bp\right)\\
 &={\mathds 1}^\dagger\frac{\partial \bD}{\partial p_j}\bp -
{\mathds 1}^\dagger\frac{\partial \bD}{\partial p_j}\bp = 0.
 \end{align*}
Hence,\footnote{It is easy to see that this property also holds for maps that are composition of maps with the structure in
\eqref{eq:class}.}
\begin{equation}\label{eq:normalization}
{\mathds 1}^\dagger \bQ^\dagger = {\mathds 1}^\dagger \bPi^\dagger = \mathds{1}^\dagger.
\end{equation}
The following serves as a key ingredient in subsequent developments.\\

\begin{thm}\label{thm1} Let  $\bbf(\cdot)$ be as in \eqref{eq:class} with
$\bD(\bp)$ continuously differentiable, and suppose that the Jacobian matrix $\bQ$ defined in \eqref{eq:oneQ} has strictly positive entries in $\cS_{n-1}^o$. The following hold:\\[-.25in]
\begin{itemize}
\item[(a)] $\bbf$ is strictly contractive in $\ell_1$ in compact subsets of $\cS_{n-1}^o$.\\[-.1in]
\item[(b)] Provided $\bbf$ has a fixed point in $\cS_{n-1}^o$, this fixed point is the only fixed point and it is globally attracting. \end{itemize}
\end{thm}
\begin{proof} Consider two probability vectors $\bp^a$ and $\bp^b$ in $\cS_{n-1}^o$, and let $\balpha:=(\bp^b-\bp^a)_+$ be the vector with the positive entries of the difference $\bp^b-\bp^a$ and $\bbeta:=-(\bp^b-\bp^a)_-$ contain the entries that originally appear with negative sign, while setting the remaining entries to be zero in both cases. Thus,
\[
\bp^b-\bp^a = \balpha - \bbeta,
\]
but in this representation $\balpha$ and $\bbeta$ have non-negative entries and have no common support, i.e., $\balpha_{i}\bbeta_{i}=0$ as they are not simultaneously $\neq 0$. Since $\mathds{1}^\dagger (\bp^b-\bp^a)=0$, it follows that $\mathds{1}^\dagger \bbeta=\mathds{1}^\dagger \balpha$, hence,
\begin{equation}\nonumber
\|\bbeta\|_1=\|\balpha\|_1=:\gamma
\end{equation}
and
\begin{align*}
\|\bp^b-\bp^a\|_1&=\sum_i |p^b_i-p^a_i|\\
&=\|\bbeta-\balpha\|_1\\
&=\sum_i\beta_i+\sum_i \alpha_i \\
&=\|\bbeta\|_1+\|\balpha\|_1\\
&=2\gamma.
\end{align*}
Now consider a path $\bp(\lambda)=(1-\lambda)\bp^a+\lambda \bp^b$ for $\lambda\in[0,1]$ and consider comparing the distance between $\bp^b$ and $\bp^a$ to the length of the path
\[
\bq(\lambda)=\bPi(\bp(\lambda))^\dagger \bp(\lambda),\;\lambda\in[0,1].
\]
Clearly,
\[
\differential\bp(\lambda)=(\balpha-\bbeta)\differential\lambda,
\]
and thus
\begin{align*}
\int_{\lambda = 0}^1 \|\differential\bp(\lambda)\|_1 &= \int_0^1\|\balpha-\bbeta\|_1 \differential\lambda\\
&= \|\balpha-\bbeta\|_1 \int_0^1 \differential\lambda = \|\bp^b-\bp^a\|_1.
\end{align*}
The entries of $\bQ$ are bounded away from zero
in any compact subset of $\cS_{n-1}^o$, hence we can assume that they are greater than $\frac{\epsilon}{n}$ along the path, for some $\epsilon>0$ which may depend on the compact subset.
Then,
\begin{align}\nonumber
\int_{\lambda = 0}^1 \|\differential\bq(\lambda)\|_1 &=\int_0^1 \| \bQ(\bp(\lambda))^\dagger (\balpha-\bbeta) \|_1 \differential\lambda\\
&\hspace*{-20pt}\leq (1-\epsilon)\int_0^1 \left(\| \bQ(\bp(\lambda))^ \dagger \bbeta\|_1 +\| \bQ(\bp(\lambda))^\dagger \balpha) \|_1\right)\differential\lambda
\label{eq:ineq}\\
&=(1-\epsilon)\int_0^1 (\|\bbeta\|_1 + \|\balpha \|_1)\differential\lambda
\label{eq:eq}\\
&=(1-\epsilon) \|\bp^b-\bp^a\|_1.\nonumber
\end{align}
To see why the inequality \eqref{eq:ineq} holds, note that for each $\lambda$,
\[
\bv^\beta:=\bQ(\bp(\lambda))^\dagger \bbeta \mbox{ and } \bv^\alpha:=\bQ(\bp(\lambda))^\dagger \balpha
\]
are vectors with positive entries, while $\|\bv^\beta\|_1=\|\bbeta\|_1=\gamma$, and $\|\bv^\alpha\|_1=\|\balpha\|_1=\gamma$ since $\bQ$ is row stochastic. The entries of $\bv^\beta$ are strictly larger than $\frac{\epsilon}{n} \|\bbeta\|_1=\frac{\epsilon \gamma}{n}$  and, similarly, the entries of $\bv^\alpha$ are strictly larger than the same value, $\frac{\epsilon}{n} \|\balpha\|_1=\frac{\epsilon \gamma}{n}$. 
 Therefore,
\begin{align*}
\|\bv^\beta-\bv^\alpha\|_1&\leq \|\bv^\beta\|_1+\|\bv^\alpha\|_1 -
2\epsilon\gamma\\ 
&=2\gamma(1-\epsilon),
\end{align*}
establishing the claimed inequality.
Finally, the metric property of $\|\cdot\|_1$ implies that
\[\|\bq(1)-\bq(0)\|_1\leq \int_{\lambda = 0}^1 \|\differential\bq(\lambda)\|_1,\]
where $\bq(1)=\bPi(\bp^b)^\dagger \bp^b$ and $\bq(0)=\bPi(\bp^a)^\dagger \bp^a$. Hence,
\[\|\bPi(\bp^b)^\dagger \bp^b-\bPi(\bp^a)^\dagger \bp^a\|_1\leq (1-\epsilon)\|\bp^b-\bp^a\|_1.
\]
This proves the first claim (part (a)). 

Now assuming that $\bbf$ has a fixed point $\bp^a$ in $\cS_{n-1}^o$, consider any other point $\bp^b\in\cS_{n-1}$ and the path $\bp(\lambda)=(1-\lambda)\bp^a+\lambda \bp^b$ for $\lambda\in[0,1]$ as before. Since $\bp^a$ is in the interior of $\cS_{n-1}$ there is an $\epsilon_1>0$ such that
${\mathcal B}_{\ell_1}(\bp,\epsilon_1):=\{\bp\in\cS_{n-1}\mid \|\bp-\bp^a\|_1\leq \epsilon_1\}$ is also in the interior of $\cS_{n-1}$. The elements of $\bQ(\bp)$ are greater than, $\frac{\epsilon_2}{n}$, for some $0<\epsilon_{2} <1$, in ${\mathcal B}_{\ell_1}(\bp,\epsilon_1)$. 
Split the path $\{\bp(\lambda)\mid \lambda\in[0,1]\}$ into two parts: $\{\bp(\lambda)\mid \lambda\in[0,\lambda_1]\}$ that is contained in ${\mathcal B}_{\ell_1}(\bp,\epsilon_1)$ and $\{\bp(\lambda)\mid \lambda\in[\lambda_1,1]\}$ that is not. The portion of the path that is in ${\mathcal B}_{\ell_1}(\bp,\epsilon_1)$ contracts when mapped via $\bbf$ by $1-\epsilon_2$, whereas the length of remaining is nonincreasing.
Thus,
	\begin{align*}
	\|\bbf(\bp^b)\hspace*{-2pt}-\hspace*{-2pt}\bbf(\bp^a)\|_1 &\le  \int_0^{\lambda_1} \|d\bq(\lambda)\|_1+ \int_{\lambda_1}^1 \|d\bq(\lambda)\|_1
	\\
	&\le (1-\epsilon_2)\int_0^{\lambda_1} \|d\bp(\lambda)\|_1+ \int_{\lambda_1}^1 \|d\bp(\lambda)\|_1
	\\
	&\le (1-\epsilon_2)\lambda_1 \|\bp^b-\bp^a\|_1+(1-\lambda_1)\|\bp^b-\bp^a\|_1
	\\
	&\le(1-\epsilon_2\lambda_1) \|\bp^b-\bp^a\|_1.
	\end{align*}
Finally, we notice that $1-\epsilon_2\lambda_1\le 1-\epsilon_2\epsilon_1/2$ since $\|\bp^b-\bp^a\|_1\le 2$.
In total, the $\ell_1$-distance between $\bp^a$ and the elements of the sequence $\bp^b$, $\bbf(\bp^b)$, $\bbf(\bbf(\bp^b))$, \ldots, reduces to zero exponentially fast with a rate of at least $1-\epsilon_2\epsilon_1/2$. This proves the second part (part (b)).
\end{proof}

\begin{remark}
We note that analogous results to Theorem \ref{thm1} for $\ell_1$-contractivity for monotone nonlinear compartmental continuous-time systems were proven in Como etal.\ \cite{como2015throughput,como2017resilient} (e.g., see \cite[Lemma 1]{como2015throughput}), and that similar ideas underlie the differential Finsler-Lyapunov framework of Forni and Sepulchre \cite{forni2014differential,forni2016differentially} as well as work on monotone and hierarchical systems \cite{coogan2017contractive,manchester2017existence,russo2013contraction}. While our approach in this paper is to derive conditions on the differential map $\bdelta \mapsto \bQ(\bp)^{\dagger}\bdelta$ so as to guarantee $\ell_{1}$-contractivity of the map $\bp \mapsto \bPi(\bp)^{\dagger}\bp$ on $\cS_{n-1}$, it would be interesting to investigate a discrete-time Finsler-Lyapunov function approach analogous to the continuous-time case known in the literature (see e.g., Theorem 1 in \cite{forni2014differential}). Thus, the objective would be to construct a Finsler-Lyapunov function $V(\bp,\bdelta) :  \cS_{n-1} \times \mathcal{T} \rightarrow \mathbb{R}_{\geq 0}$ for the augmented map
$\left(\begin{array}{cc}\bp,& \bdelta\end{array}\right)\mapsto \left(\begin{array}{cc}\bPi(\bp)^{\dagger}\bp,& \bQ(\bp)^{\dagger}\bdelta\end{array}\right)$
so as to guarantee $\ell_{1}$-contractivity of the map $\bp \mapsto \bPi(\bp)^{\dagger}\bp$.
\\
\end{remark}

\begin{remark}
If $\bbf$ is a general nonlinear map, the Jacobian matrix $\bQ$ may fail to be stochastic for two reasons. First,
the elements of $\bQ$ may fail to be non-negative. Second, the normalization \eqref{eq:normalization} may fail unless $\bPi$ has a particular structure, as for instance the one in \eqref{eq:class}. A simple example to demonstrate the failure of \eqref{eq:normalization} is
\[
\left(\begin{array}{c}p_1\\p_2\end{array}\right)\mapsto
\left(\begin{array}{cc}p_1 & p_2\\p_2 & p_1\end{array}\right)
\left(\begin{array}{c}p_1\\p_2\end{array}\right).
\]
For this example one can readily see that $\mathds{1}^\dagger \bQ^\dagger=2\mathds{1}^\dagger\neq \mathds{1}^\dagger$.\\
\end{remark}

\begin{remark} At times it is easy to ensure that $\bbf$ in Theorem \ref{thm1} has a fixed point in the interior of $\cS_{n-1}$.
For instance, if $\bC_0=\bI$ is the identity matrix, and since $\bC^\dagger_1(\bI-\bD)\bp=(\bI-\bD)\bp$ and $\bC_1$ corresponds to a simply connected aperiodic chain, $(\bI-\bD)\bp$ is the corresponding Frobenius-Perron eigenvector and therefore lies in the interior of $\cS_{n-1}$. Conclusions can be drawn for $\bp$, accordingly, depending on $\bD$.\\
\end{remark}

\begin{cor}\label{cor1}
Let $\bPi(\bp)$ be row-stochastic and differentiable in $\bp$, and suppose that the Jacobian of the map $\bbf(\cdot)$ in \eqref{eq:nonlinearevolution} has non-negative entries.
Then, $\bbf$ is contractive (but not necessarily strictly contractive) in the $\ell_1$-metric. 
\end{cor}

Stronger statements that build on the theorem are stated next.
 All results hold for functional forms that are more general than the exponential and linear models considered in this paper.\\

%

\begin{prop}\label{cor3}
Let matrix $\bPi(\bp)$ be row-stochastic and continuously differentiable in $\bp$. Suppose $\bbf$ has a fixed point $\bp^{\star}$ in $\cS_{n-1}^o$, and for a suitable integer $m$ the differential (Jacobian) of the $m$th iterant
\begin{equation}\label{eq:iterant}
\bbf^m(\bp):=\overbrace{\bbf (\bbf(\ldots \bbf}^m(\bp)))
\end{equation}
has strictly positive entries for all $\bp\in\cS_{n-1}^o$.
Then, $\bp^{\star}$ is the unique fixed point of $\bbf$ and it is globally stable.
\end{prop}

\begin{proof}
By assumption, $\bp^{\star}$ is a fixed point of $\bbf^m$ since $\bbf^m(\bp^{\star}) = \bbf^{m-1}(\bp^{\star}) = \cdots = \bp^{\star}$. Now applying Theorem \ref{thm1} to $\bbf^m$ we conclude that $\bp^{\star}$ is the unique fixed point of $\bbf^m$ and is globally stable. Therefore, $\bp^{\star}$ is also a unique fixed point of $\bbf$, and the global stability of $\bp^{\star}$ with respect to $\bbf$ follows from the continuity of $\bbf$.
\end{proof}


\begin{prop}\label{cor4}
Let matrix $\bPi(\bp)$ be row-stochastic and continuously differentiable in $\bp$. Suppose that $\bp^\star$ in $\cS_{n-1}^o$ is a fixed point of $\bbf$ in \eqref{eq:nonlinearevolution} and that, for a suitable integer $m$,  the $m$th power 
\[ \left( \differential\bbf |_{\bp^\star} \right)^m
\]
of the Jacobian of $\bbf$ evaluated at $\bp^\star$
has strictly positive entries.
Then $\bp^\star$ is a locally attractive equilibrium.
\end{prop}

\begin{proof}
The expression $\left( \differential\bbf |_{\bp^\star} \right)^m$ is precisely the Jacobian of the $m^{\rm th}$ iterate, i.e.,
\[
\left( \differential\bbf |_{\bp^\star} \right)^m = \differential \overbrace{\,\bbf (\bbf(\ldots \bbf))}^m|_{\bp^\star}.
\]
By continuity, the entries of $\differential \bbf^m$, with $\bbf^m:= \overbrace{\,\bbf (\bbf(\ldots \bbf))}^m$, will remain positive in a neighborhood of $\bp^\star$.
It is then clear that $\bbf^m$, which is stochastic and satisfies the conditions of  Theorem \ref{thm1}, has $\bp^\star$ as a (locally) attractive fixed point. Therefore, using the continuity of $\bbf$, we conclude that $\bp^\star$ is a locally attractive fixed point for $\bbf$.
\end{proof}

\begin{prop}\label{cor5}
Let matrix $\bPi(\bp)$ be row-stochastic and continuously differentiable in $\bp$, and assume that $\bp^i$, for $i=0,1,2,\ldots,m-1$, is a periodic orbit for $\bbf$ in \eqref{eq:nonlinearevolution}, i.e., 
\[
\bp^{(i+1){\rm mod}(m)}=\bbf(\bp^{(i){\rm mod}(m)}).
\]
Suppose that the product of the Jacobians
\[ \left( \differential\bbf |_{\bp^{(i+m){\rm mod}(m)}} \right)\ldots \left( \differential\bbf |_{\bp^{(i){\rm mod}(m)}} \right)
\]
has strictly positive entries for some $i$.
Then, the periodic orbit is locally attractive.
\end{prop}

\begin{proof}
Under the stated condition, for any $i$, $\bp^i$ is a locally attractive fixed point for the $m$th iterant,
$\overbrace{\,\bbf (\bbf(\ldots \bbf))}^m|_{\bp^i}$. The fact that the orbit is locally attractive now follows from the continuity of $\bbf$.
\end{proof}

We provide next a bound on the induced $\ell_1$-incremental gain of stochastic maps in terms of the induced $\ell_1$-gain
of the Jacobian
\begin{align*}
\|\differential\bbf|_\cT\|_{(1)}&:= \max \{ \|\bQ^\dagger \bdelta\|_1 \mid \mathds{1}^\dagger \bdelta=0, \|\bdelta\|_1=1\}.
\end{align*}
This strengthens substantially the applicability of the framework since it relaxes the positivity requirement on the Jacobian, albeit this relaxed condition may be more challenging to verify globally.\\

\begin{thm}\label{thm6} Let $\bbf$ be a differentiable stochastic map as in \eqref{eq:nonlinearevolution} and as before the Jacobian $\differential\bbf(\bp)|$ is represented by a matrix $\bQ(\bp)^\dagger$. For any $\bp^b,\bp^a\in\cS_{n-1}$,
\[
\| \bbf(\bp^b)-\bbf(\bp^a)\|_1 \leq \max_{\bp\in\cS_{n-1}} \|\differential\bbf(\bp)|_\cT\|_{(1)} \|\bp^b-\bp^a\|_1,
\]
and, in general,
\begin{align}\label{eq:dobrushin}
\|\differential\bbf|_\cT\|_{(1)}&= \frac{1}{2}\max_{j,k}\sum_{i=1}^n |(\bQ(\bp))_{ji}-(\bQ(\bp))_{ki}|.
\end{align}
\label{PropositionJacobian}
\end{thm}

\begin{proof} The first claim is straightforward since, with $\balpha,\bbeta$ as in the proof of Theorem \ref{thm1},
\begin{align*}
\| \bbf(\bp^b)-\bbf(\bp^a)\|_1 &\leq \int_{\lambda = 0}^1 \|\differential\bq(\lambda)\|_1\\
&=\int_0^1 \| \bQ(\bp(\lambda))^\dagger (\bbeta-\balpha) \|_1 \differential\lambda\\
&
\leq \left(\max_{\bp\in\cS_{n-1}} \|\differential\bbf(\bp)|_\cT\|_{(1)}\right)\int_0^1 \| \bbeta-\balpha \|_1 \differential\lambda\\
&=\left(\max_{\bp\in\cS_{n-1}} \|\differential\bbf(\bp)|_\cT\|_{(1)}\right) \|\bp^b-\bp^a\|_1.
\end{align*}
Any $\bdelta\in\cT$ with $\|\bdelta\|_1=1$ can be written as $\bdelta=\frac12(\bbeta-\balpha)$ with $\balpha,\bbeta$ having nonnegative entries and $\|\balpha\|_1=\|\bbeta\|_1=1$, as before, and at any given $\bp\in\cS_{n-1}$,
\begin{align}\nonumber
\|\differential\bbf|_\cT\|_{(1)}&=\max\{ \|\bQ^\dagger(\bp)\bdelta\|_1\mid
\bdelta\in\cT\}\\\label{eq:l1induced}
&=\frac12 \max\{ \|\bQ^\dagger\bbeta-\bQ^\dagger\balpha\|_1\mid
\balpha,\bbeta\in\cS_{n-1}\}.
\end{align}
The claim follows by convexity. To see this, note that $\|\bQ^\dagger \bbeta - {\bm \nu}\|_1$, for $\bm \nu$ constant, is a convex function of $\bbeta\in\cS_{n-1}$. Therefore the maximal value will be attained at an extreme point, i.e., a vertex, and likewise when maximizing with respect to $\balpha$. Thus, the extremal will be at a point where both $\bbeta$ and $\balpha$ have a single nonzero element (and thereby select a corresponding row of $\bQ$).\end{proof}

We note that the expression \eqref{eq:dobrushin} for the induced $\ell_1$-norm of linear maps is the so-called Markov-Dobrushin coefficient of ergodicity \cite{markov1906extension,dobrushin1956central,seneta2006markov,gaubert2015dobrushin}
that characterizes the contraction rate of Markov operators with respect to this norm (also, total variation).
For nonlinear operators on probability simplices (nonlinear Markov Chains), the same is true.
The above propositions provide candidate certificates for stability of equilibria $\bp^\star$ and highlight the fact that
the $\ell_1$-distance is a natural Finsler-Lyapunov function  in the sense of Forni and Sepulchre \cite{forni2014differential}.
The essence is that $\ell_1$-contractivity of the nonlinear dynamics $\bp_{\rm next}=\bbf(\bp)$, and stability of fixed points or periodic orbits, may be deduced from the infinitesimal properties of $\bbf$ in the $\ell_1$-metric. The approach is illustrated in the next sections.

\section{Exponential-influence models}\label{sec:Gs}
In this section we analyze the model in \eqref{eq:lazywalk} for the cases where the {\em reinforcement} function $r(x)$ is either
$1-e^{-\gamma x}$ or $e^{-\gamma x}$, for some $\gamma > 0$. The first choice satisfies $r(0)=0$ and $r^\prime(0)=\gamma$, and thereby strengthens the return probabilities\footnote{Naturally, the rates also depend on the choice of $\bC$.} for states with relatively large marginal probability at corresponding times $t$. The second choice has $r(0)=1$ and $r^\prime(0)=-\gamma$, has the tendency to do the opposite.

Throughout we assume that $\bC$ is an irreducible aperiodic row-stochastic matrix, and we denote by $\bc$ the unique (positive) Frobenius-Perron left eigenvector, i.e., $\bc$ satisfies
\[
\bC^\dagger \bc = \bc.
\]
It is normalized so that $\mathds{1}^\dagger \bc=1$ and, because of the irreducibility assumption, $\bc$ has positive entries.

\subsection{\bf Case $\mathbf{ r(x) = 1-e^{-\gamma x}}$ for $\mathbf{0<\gamma \leq 1}$.}
\begin{prop}
\label{theorem4}\label{thm:thm4}
With $\bC$ as above and for any $\gamma\in(0,1]$ consider the map \begin{subequations}\label{eq:lazywalk1gamma}
\begin{align}
&\bp(t)\mapsto \bbf(\bp(t))=\bp(t+1),\mbox{ where}\\
&\bbf(\bp(t)) = \left(\diag(\mathds{1}-e^{-\gamma \bp(t)})+\bC^\dagger\diag(e^{-\gamma \bp(t)})\right) \bp(t).
\end{align}
\end{subequations}
The following hold:
\spacingset{.7}
\begin{itemize}
\item[i)] $\bbf(\cdot)$ is contractive in $\ell_1$,\\
\item[ii)] $\bbf$ has a unique fixed point $\bp^\star$ with entries satisfying
$
e^{-\gamma p^{\star}_{i}}p^{\star}_{i} = \kappa c_{i}$,
for some $\kappa>0$,\\
\item[iii)] starting from an arbitrary $\bp(0)\in\cS_{n-1}$,
\[\bp^\star=\lim_{t\to\infty}\bp(t).\]
\end{itemize}
\spacingset{1}
\end{prop}

\begin{proof}
The Jacobian $\differential\bbf$ is of the form
\begin{align*}
\bdelta\mapsto &\left(\diag(\br(\bp) + \bp\odot \br^\prime(\bp)) \right.\\
&\left.+ \bC^\dagger (\bI-\diag(\br(\bp)+\bp\odot \br^\prime(\bp)))\right)\bdelta\\
&\hspace*{-20pt}= \underbrace{\left( \diag(\mathds{1} - e^{-\gamma \bp} + \gamma\bp\odot e^{-\gamma\bp}) + \bC^\dagger \diag(e^{-\gamma \bp} - \gamma \bp\odot e^{-\gamma \bp}) \right)}_{\bQ(\bp)^\dagger}\bdelta,
\end{align*}
where $\odot$ denotes the entry-wise multiplication of vectors, and for a vector $\bv=(v_{i})_{i=1}^n$, $e^{\bv}$ denotes the vector with entries $e^{v_{i}}$. Since both functions $1-e^{-\gamma x}+\gamma xe^{-\gamma x}$ and $e^{-\gamma x}-\gamma xe^{-\gamma x}$ take non-negative values on $[0,1]$,
$\bQ(\bp)^\dagger$ is a (column) stochastic matrix. Thus, $\bbf$ is contractive. 

Any fixed point of $\bbf$ must satisfy
\begin{equation}\label{eq:fixedpoint}
\bp=\left(\diag(1-e^{-\gamma \bp}) + \bC^\dagger (e^{-\gamma \bp}))\right) \bp.
\end{equation}
Rearranging terms we see that $\bp e^{-\gamma \bp}$ is proportional to $\bc$ (the Frobenius-Perron vector of $\bC$), and therefore,
\begin{equation}\label{eq:system}
p_{i}e^{-\gamma p_{i}}=\kappa c_{i}, \; i=1,\ldots,n.
\end{equation}
The function $xe^{-\gamma x}$ is monotonic on $[0,1]$ and hence, for any 
\[
\kappa \leq \frac{1/(\gamma e)}{\displaystyle\max_{i}\{c_{i}\}}=:\kappa_{\max},
\]
there is a unique solution $\{p_{i}\mid i=1,\ldots,n\}$ of \eqref{eq:system}.
Let now
$
s(\kappa):=\sum_i p_{i}.
$
The function $s(\kappa)$ is monotonically increasing as a function of $\kappa$ and has $s(0)=0$. For
$\kappa=\kappa_{\max}$ one of the $p_{i}$'s is equal to $1$ and hence $s(\kappa_{\max})\geq 1$. Thus,
the equation $s(\kappa)=1$ has a unique solution that corresponds to the probability vector $\bp^\star$ that satisfies \eqref{eq:fixedpoint}. Thus the fixed point $\bp^\star$ is unique.

Further, $\bQ$ inherits irreducibility from $\bC^\dagger$ in $\cS_{n-1}^o$ since it has the same pattern of positive entries; in addition it is aperiodic, irrespective of $\bC$, because its diagonal is not zero. Hence, independently of $\bp$, there exists integer $m$ such that
\begin{equation}\label{eq:differential_fff}
\bQ(\,\overbrace{\bbf(\ldots \bbf(\bp))}^{m-1}\,)^\dagger\ldots \bQ(\bbf(\bp))^\dagger \bQ(\bp)^\dagger
\end{equation}
has all entries positive. The expression in \eqref{eq:differential_fff} is precisely the differential of the $m$th iterant (cf.\ \eqref{eq:iterant}). 
By Proposition \ref{cor3}, $\bp^\star$ is globally attractive.  
\end{proof}

%

\begin{remark}\label{RemarkFromExp2DF}
More in the style of DeGroot-Friedkin models \cite{degroot1974reaching,friedkin2010attitude} of the general form (\ref{eq:DeGrootFriedkin}), one may consider a model
\begin{equation}\label{eq:lazywalkFP}
\bp(t+1) = \left[\diag(\mathds{1}-e^{-\gamma \bp(t)})+\bC^\dagger\diag(e^{-\gamma \bp(t)}))\right]_{\rm FP}.
\end{equation}
Then, $\diag(\mathds{1}-e^{-\gamma\bp(t)})+\bC^\dagger\diag(e^{-\gamma\bp(t)}))$ is irreducible and therefore, an alternative formula for $\bp(t+1)$ is
\[
\bp(t+1)=\lim_{m\rightarrow\infty} \left(\diag(\mathds{1}-e^{-\gamma\bp(t)})+\bC^\dagger \diag(e^{-\gamma\bp(t)}))\right)^{m} \bp(t).
\]
Comparing with \eqref{eq:lazywalk1gamma}, the fixed point $\bp^\star$ in Proposition \ref{thm:thm4} is also a fixed point for \eqref{eq:lazywalkFP}.
\end{remark}

\subsection{\bf Case $\mathbf{r(x)=1-e^{-\gamma x}}$ for $\mathbf{\gamma> 1}$.}

The case $\gamma>1$ is substantially different. Here, there can be several attractive points of equilibrium for the nonlinear dynamics in \eqref{eq:lazywalk1gamma} and even more complicated nonlinear behavior. 
In fact, we suggest that {\em such a behavior may be more appropriate for models of opinion dynamics} as it is reasonable to expect a different outcome depending on the starting point (that encapsulates confidence/beliefs of individuals). We illustrate the behavior with two numerical examples for $3$-state Markov chains to highlight differences with the case when $\gamma\leq 1$.
\subsubsection{Example}
We consider the dynamics in \eqref{eq:lazywalk1gamma}
for a $3$-state Markov chain (i.e., $n=3$) with $\gamma=4$ and
\begin{eqnarray}
\bC=\left[\begin{matrix}0.8 & 0.1 & 0.1\\0.4 & 0.2 & 0.4\\ 0.4 & 0.4 & 0.2\end{matrix}\right].
\label{CmatrixExampleSec3B1}
\end{eqnarray}
The left Frobenius-Perron eigenvector of $\bC$ is $(2/3,1/6,1/6)^\dagger$. The fixed-point conditions for possible stationary distributions become
\begin{eqnarray*}
&&e^{-4p^{\star}_{1}}p^{\star}_{1} = \kappa \frac{2}{3},\\
&&e^{-4p^{\star}_{2}}p^{\star}_{2} = \kappa \frac{1}{6},\\
&&2p^{\star}_{2}+ p^{\star}_{1} = 1.
\end{eqnarray*}
Upon eliminating $\kappa$ between the first two, and substituting $p_1$ in terms of $p_2$, we obtain
\begin{eqnarray}
\frac{1 - 2 p^{\star}_{2}}{p^{\star}_{2}} e^{-4(1 - 3 p^{\star}_{2})} = 4.
\end{eqnarray}
This equation has the unique solution
\[
\bp^{\star}:=(0.9904,0.0048,0.0048)^\dagger.
\]
It turns out that this is a locally attractive fixed point. This can be verified by evaluating the Jacobian of $\bbf$ at $\bp^{\star}$ as
\[
\differential\bbf|_{\bp^\star} = \left[\begin{matrix}1.0113 & 0.3849 & 0.3849\\-0.0056 & 0.2303 & 0.3849\\ -0.0056 & 0.3849 & 0.2303\end{matrix}\right].
\]
Even though the Jacobian has negative entries it is still strictly contractive. Indeed, we explicitly evaluate the induced gain using Theorem \ref{PropositionJacobian} and this is
\begin{eqnarray*}
\|\differential\bbf|_\cT\|_{(1)}= \frac{1}{2} \max\{1.2528, 1.2528, 0.3092\} = 0.6264 < 1.
\end{eqnarray*}
Thus $\bp^{\star}$ is a stable fixed point. This analysis is consistent with simulations shown in Fig. \ref{gamma4n3_example1}. In the figure we depict trajectories (in different color) starting from random initial conditions that clearly tend to $\bp^{\star}$. 
\begin{center}
\begin{figure}[htb]
  \centering
    \includegraphics[width=0.85\linewidth]{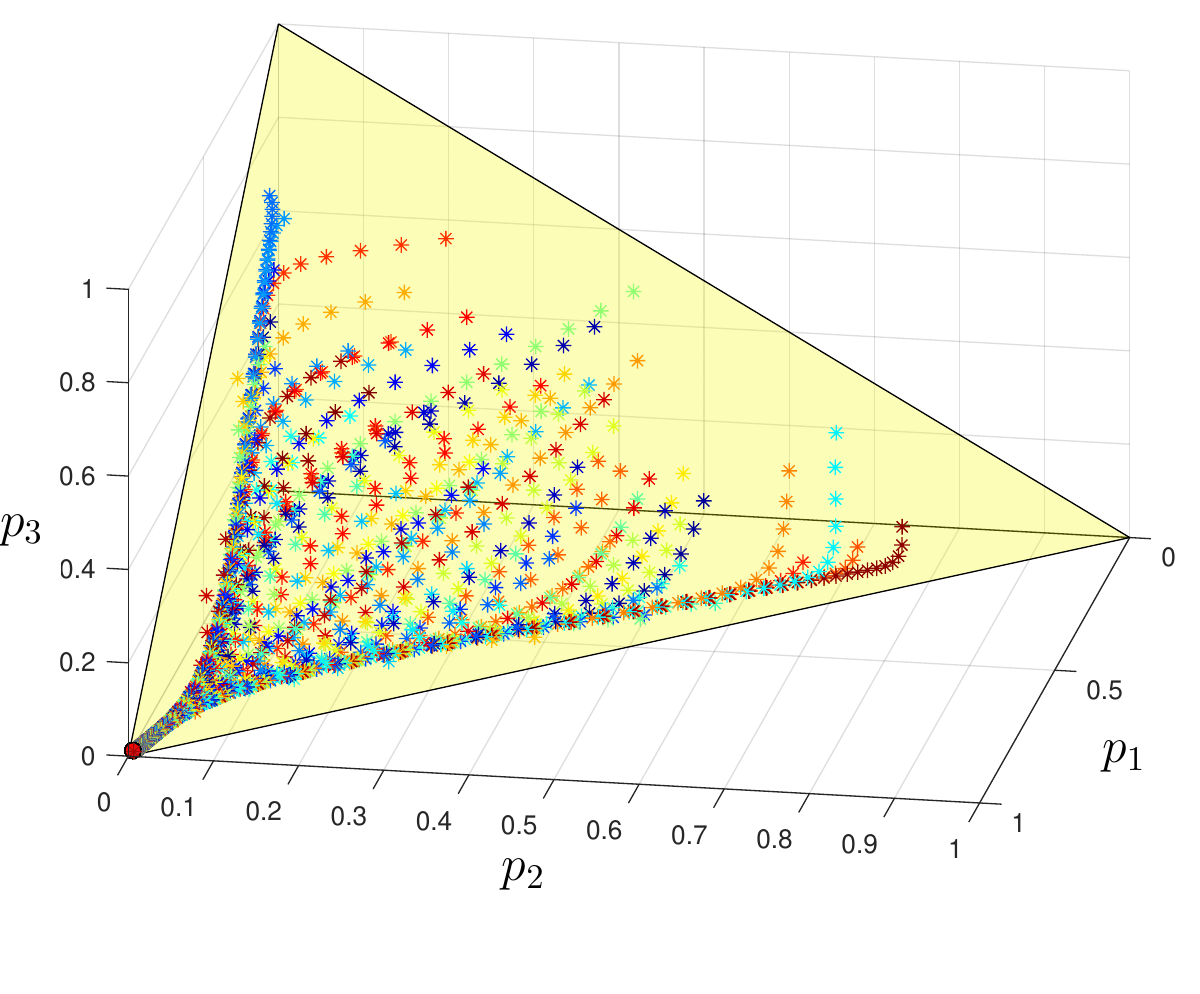}
    \caption{Convergence of trajectories to a unique fixed point for the $3$-state exponential model \eqref{eq:lazywalk1gamma} with $\gamma = 4$ and influence matrix $\bC$ given by (\ref{CmatrixExampleSec3B1}).}
    \label{gamma4n3_example1}
\end{figure} 
\end{center}

\subsubsection{Example}
Once again we consider a $3$-state Markov chain with $\gamma=4$, but this time we take
\begin{eqnarray}
\bC=\left[\begin{matrix}0 & 0.5 & 0.5\\0.5 & 0 & 0.5\\ 0.5 & 0.5 & 0\end{matrix}\right].
\label{CmatrixExampleSec3B2}
\end{eqnarray}
The fixed-point equations have $7$ solutions (taking into account symmetries). Out of those, three are attractive fixed points
with coordinates cyclically selected from $\{1-a,a/2,a/2\}$ for $a = 0.046$. The remaining four
are unstable fixed points. One is at the center $(1/3,1/3,1/3)^\dagger$ (due to symmetry), and the rest have coordinates cyclically selected from $\{1-a,a/2,a/2\}$ for $a = 0.874$. Just like the previous example, we can verify stability by computing the Jacobian $\differential\bbf$ at fixed points. For instance, for the fixed point $\bp^{\star}_{a} = (0.954,0.023,0.023)^\dagger$, we have
\[
\differential\bbf |_{\bp^{\star}_{a}} =\left[\begin{matrix}1.0620 & 0.4141 & 0.4141\\-0.0310 & 0.1718 & 0.4141\\ -0.0310 & 0.4141 & 0.1718\end{matrix}\right],
\]
and  
\begin{eqnarray*}
\|\differential\bbf|_\cT\|_{(1)}= \frac{1}{2} \max\{1.2958, 1.2958, 0.4846\} = 0.6479 < 1.
\end{eqnarray*}
Applying Theorem \ref{PropositionJacobian}, we conclude that $\bp^{\star}_{a}$ is a stable fixed point. For another fixed point $\bp^{\star}_{b} = (0.1260,0.4370,0.4370)^\dagger$, we have
\[
\differential\bbf |_{\bp^{\star}_{b}} =\left[\begin{matrix}0.7004 & -0.0651 & -0.0651\\0.1498 & 1.1302 & -0.0651\\ 0.1498 & -0.0651 & 1.1302\end{matrix}\right],
\]
and
\begin{eqnarray*}
\|\differential\bbf|_\cT\|_{(1)}= \frac{1}{2} \max\{1.9608, 1.9608, 2.3907\} = 1.1954 > 1.
\end{eqnarray*}
Numerical evidence shown in Fig. \ref{gamma4n3_example2} confirms that $\bp^{\star}_{a}$ is stable and $\bp^{\star}_{b}$ is unstable. Convergence of trajectories depends on the initial conditions with respect to the basins of attraction for the three stable fixed points. The qualitative behavior of the trajectories around the four unstable and three stable fixed points is illustrated in Fig. \ref{ExpGamSevenFixedPoints}. 
\begin{center}
\begin{figure}[ht]
  \centering
    \includegraphics[width=0.85\linewidth]{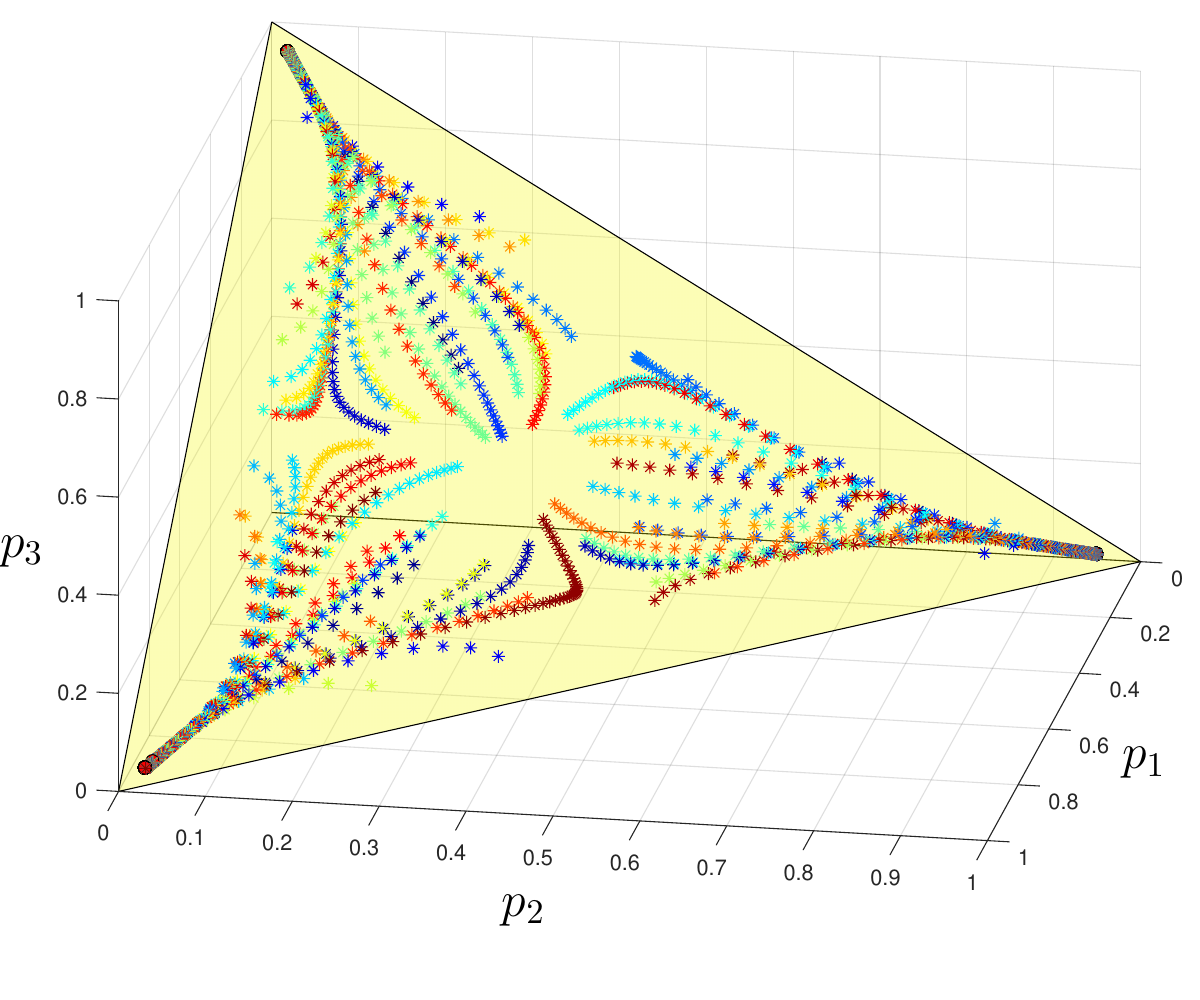}
    \caption{For the $3$-state exponential model \eqref{eq:lazywalk1gamma} with $\gamma = 4$ and influence matrix $\bC$ given by (\ref{CmatrixExampleSec3B2}), trajectories converge to one of the three stable fixed points.}
    \label{gamma4n3_example2}
\end{figure} 
\end{center}
\begin{center}
\begin{figure}[ht]
  \centering
    \includegraphics[width=0.7\linewidth]{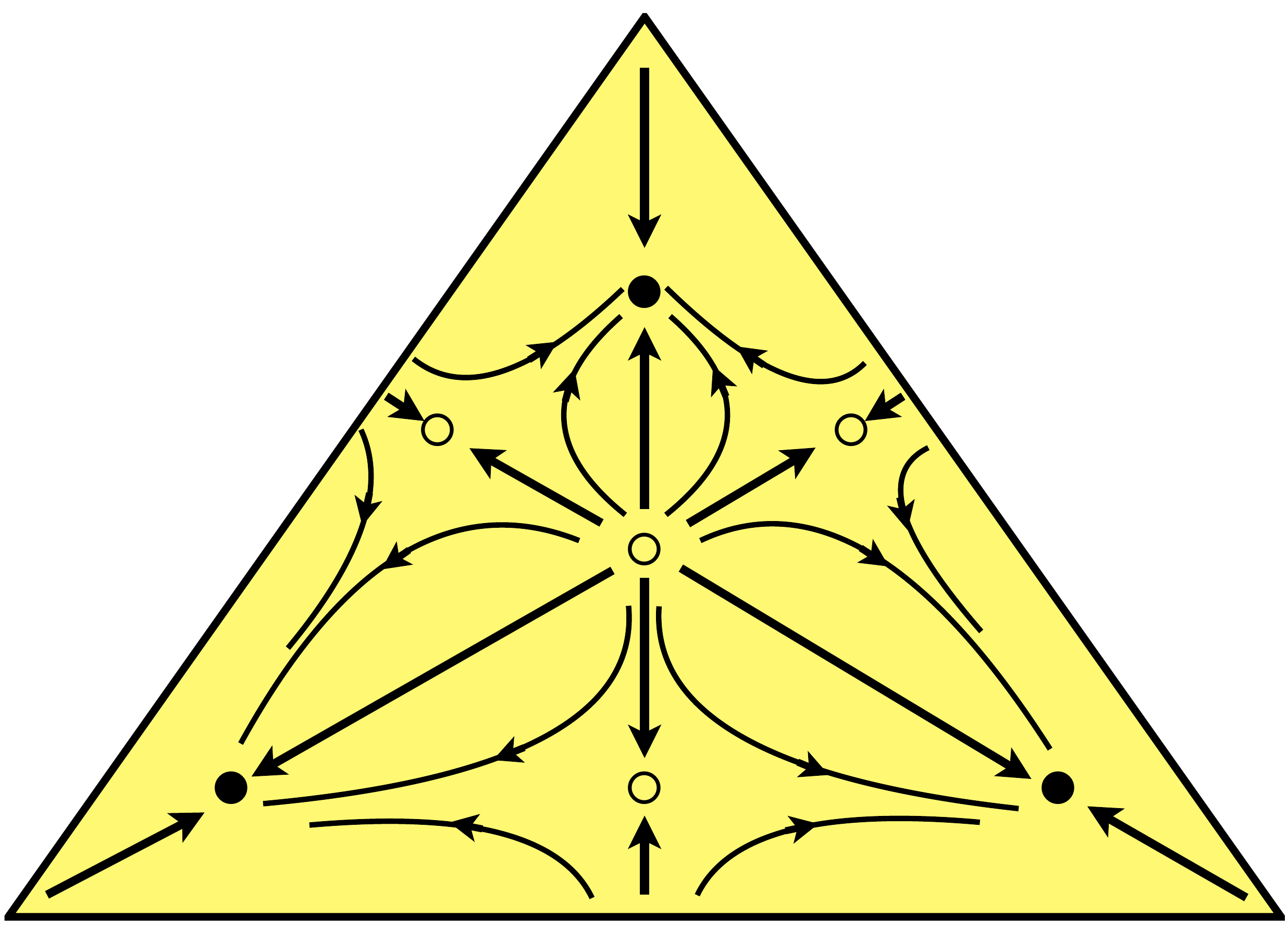}
    \caption{The qualitative behavior of dynamics \eqref{eq:lazywalk1gamma} with $\gamma >1$ as observed in Fig. \ref{gamma4n3_example2}, where three stable fixed points (solid circles) and four unstable fixed points (empty circles) coexist on the simplex.}
    \label{ExpGamSevenFixedPoints}
\end{figure} 
\end{center}

\subsection{\bf Case $\mathbf{r(x)=e^{-\gamma x}}$ for $\mathbf{\gamma\leq 1}$.}
In this case there is a unique fixed point and it is always globally attractive.
We summarize our conclusions as follows:
\begin{prop} \label{repellingresult}
For any $\gamma \in[0,1]$ consider
\begin{subequations}\label{repellingwalk}
\begin{align}
&\bp(t)\mapsto \bbf(\bp(t))=\bp(t+1),\mbox{ where}\\
&\bbf(\bp(t)) = \left(\diag(e^{-\gamma\bp(t)})+\bC^\dagger \diag(\mathds{1}-e^{-\gamma\bp(t)})\right) \bp(t).\end{align}
\end{subequations}
The map $\bbf$ is contractive in $\ell_1$ and, starting from an arbitrary $\bp(0)\in\cS_{n-1}$, the limit $\bp^\star=\lim_{t\to\infty}\bp(t)$ exists, is unique, and its entries satisfy
$
\left(1-e^{-\gamma p^{\star}_{i}} \right) p^{\star}_{i}=\kappa c_{i},
$
for some $\kappa>0$.
\end{prop}
\begin{proof}
First, the Jacobian matrix $\bQ(\bp)^\dagger$ is of the form
\begin{align*}
\diag(e^{-\gamma\bp}-\gamma\bp\odot e^{-\gamma\bp})+\bC^\dagger \diag(\mathds{1}-e^{-\gamma\bp}+\gamma\bp\odot e^{-\gamma\bp}).
\end{align*}
Notice that $\bQ(\bp)^\dagger$ is differentiable in $\bp$, and for $\gamma\leq 1$, is a (column) stochastic matrix with non-negative entries. Therefore, by Corollary \ref{cor1}, the map (\ref{repellingwalk}) is contractive in $\ell_1$ and inherits irreducibility from $\bC^{\dagger}$ in $\cS_{n-1}^o$. Following a similar line of argument as in Proposition \ref{theorem4}, uniqueness of the fixed point for map \eqref{repellingwalk} is guaranteed. Next, we write the stationarity conditions
\[
\bp^{\star}=\left(\diag\left(e^{-\gamma\bp^{\star}}\right)+\bC^\dagger\diag\left(\mathds{1}-e^{-\gamma\bp^{\star}}\right)\right)\bp^{\star},
\]
equivalently, 
\[
\left(\mathds{1}-e^{-\gamma\bp^{\star}}\right)\odot\bp^{\star}=\bC^\dagger\left(\mathds{1}-e^{-\gamma\bp^{\star}}\right)\odot\bp^{\star},
\]
to obtain that  
 \begin{align}
 \left(1-e^{-\gamma p^{\star}_{i}}\right)p^{\star}_{i}=\kappa c_{i}, \quad  i=1,\ldots,n,
 \label{FixedptEqnExpRepel}
 \end{align}
where $c_{i}$ denotes the $i$-th entry of the Frobenius-Perron vector of $\bC$ and
$\kappa=\sum_{i=1}^{n}\left(1-e^{-\gamma p^{\star}_{i}}\right)p^{\star}_{i}$.
\end{proof}

\subsection{\bf Case $\mathbf{r(x)=e^{-\gamma x}}$ for $\mathbf{\gamma> 1}$.}
In this case too there exists a unique fixed point in any dimension (any $n$).
This follows easily as the fixed-point conditions are the same,
\begin{align*}
\left(1-e^{-\gamma p^{\star}_{i}} \right) p^{\star}_{i}=\kappa c_{i}.
\end{align*}
Then, for all $\gamma>0$, $(1-e^{-\gamma x})x$ is a monotonically increasing starting at $0$ for $x=0$. Solving for a given $\kappa$, the sum $\sum_{i=1}^n p^{\star}_{i}(\kappa)$ is also monotonically increasing function of $\kappa$ and its value exceeds $1$ for a suitable $\kappa$. Thus, there is a unique solution $p^{\star}_{i}(\kappa)$ which is a probability vector (and the $p^{\star}_{i}$'s sum up to $1$).

However, interestingly, the nonlinear dynamics now display diverse behaviors.
Below we give three examples. In the first two the unique fixed point is attractive, but they differ, in that assurances for stability are drawn (for the second example) by computing the norm of the differential of higher iterants ($2$nd in this case). In the third example we observe a $2-$periodic attractive orbit.

\begin{center}
\begin{figure}[t]
  \centering
    \includegraphics[width=0.87\linewidth]{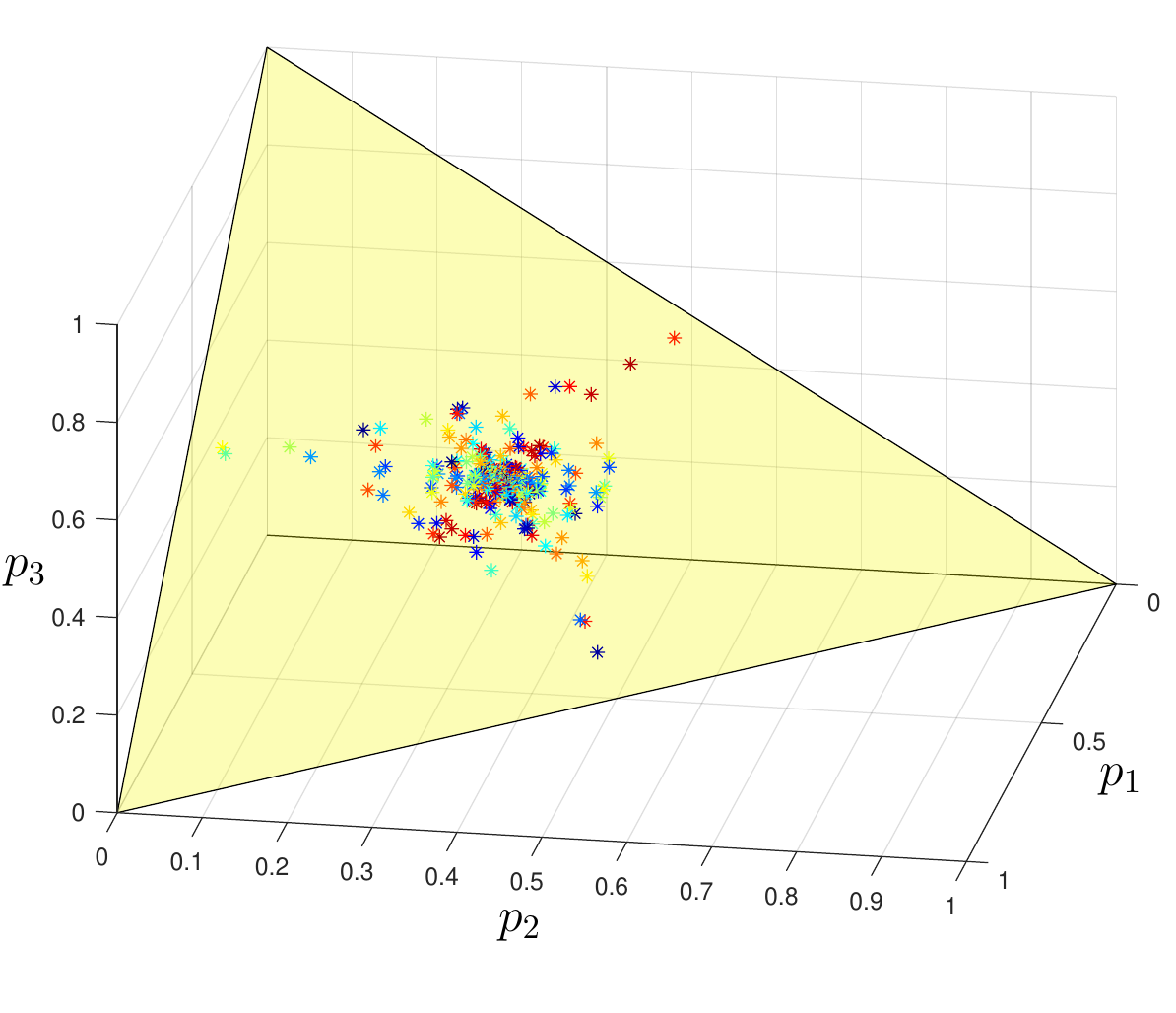}
    \caption{For the $3$-state exponential model \eqref{repellingwalk} with $\gamma = 4$ and influence matrix $\bC$ given by (\ref{CmatrixExampleSec3D1}), trajectories converge to the unique stable fixed point $\bp^{\star} = (1/3,1/3,1/3)^{\dagger}$.}
    \label{ExpRepellingExample1}
\end{figure} 
\end{center}

\subsubsection{Example}
We consider a 3-state Markov chain with $\gamma=4$, and 
\begin{eqnarray}
\bC=\left[\begin{matrix}0 & 0.5 & 0.5\\0.5 & 0 & 0.5\\ 0.5 & 0.5 & 0\end{matrix}\right].
\label{CmatrixExampleSec3D1}
\end{eqnarray}
Since $\bC$ is doubly stochastic, the unique fixed point for \eqref{FixedptEqnExpRepel} is $\bp^{\star} = (1/3,1/3,1/3)^{\dagger}$, and we have
\[
\differential\bbf |_{\bp^{\star}} =\left[\begin{matrix}-0.0880 & 0.5440 & 0.5440\\0.5440 & -0.0880 & 0.5440\\ 0.5440 & 0.5440 & -0.0880\end{matrix}\right],
\]
and  
\begin{eqnarray*}
\|\differential\bbf|_{\bp^{\star}}\|_{(1)}=\frac{1}{2} \max\{1.2640, 1.2640, 1.2640\}=0.6320 < 1.
\end{eqnarray*}
Using Theorem \ref{PropositionJacobian}, we conclude that $\bp^{\star}$ is a stable fixed point.

\subsubsection{Example}
For $\gamma=4$, now take 
\[
\bC=\left[\begin{matrix}0 & 0 & 1\\0.5 & 0.5 & 0\\ 0.5 & 0.5 & 0\end{matrix}\right].
\]
The unique fixed point is again $\bp^{\star} = (1/3,1/3,1/3)^{\dagger}$.
Here,
\[
\differential\bbf |_{\bp^{\star}} =\left[\begin{matrix}-0.0880 & 0.5440 & 0.5440\\
    0    & 0.4560   & 0.5440\\
    1.0880   &      0 &  -0.0880
    \end{matrix}\right],
\]
and  
\begin{eqnarray*}
\|\differential\bbf|_{\bp^{\star}}\|_{(1)}= 1.1760.
\end{eqnarray*}
However,
\begin{eqnarray*}
\|\differential\bbf^2|_{\bp^{\star}}\|_{(1)}= 0.7911.
\end{eqnarray*}
This ensures local attractiveness.

\subsubsection{Example}
Once again we consider a 3-state Markov chain with $\gamma=4$, but we now take 
\begin{eqnarray}
\bC=\left[\begin{matrix}0 & 0 & 1\\0.8 & 0 & 0.2 \\0.8 & 0.2& 0 \end{matrix}\right].
\label{periodicmatrixc2}
\end{eqnarray}
Uniqueness of a fixed point is  guaranteed. This turns out to be
\[\bp^{\star}=(0.4173,    0.1537,    0.4298)^\dagger.
\]
It turns out that
\[\differential\bbf |_{\bp^{\star}} =
  \left[\begin{matrix}-0.1261&    0.6333&    0.9031\\
         0 &   0.2084&    0.2258\\
    1.1261&    0.1583 &  -0.1289
\end{matrix}\right]
\]
has $\ell_1$-norm equal to $1.255$, and so do the differentials of higher order iterants.
However, a stable $2$-periodic orbit now appears
alternating between
\[
\bp^{a}\!=\!(0.1943, 0.1042, 0.7015)^\dagger \mbox{ and }\bp^{b}\!=\!(0.6450, 0.2005, 0.1545)^\dagger.
\]
The periodic orbit is locally attractive. The Jacobians at these two points are
\[
  \differential\bbf |_{\bp^a} = \left[\begin{matrix}   0.1024    &0.4923    &0.8873\\
         0   & 0.3846 &   0.2218\\
    0.8976  &  0.1231 &  -0.1092
      \end{matrix}\right]
 \]   
    and 
    \[
     \differential\bbf |_{\bp^b} =  \left[\begin{matrix}    -0.1197  &  0.7290  &  0.6352\\
         0   & 0.0888 &   0.1588\\
    1.1197  &  0.1822 &   0.2060
    \end{matrix}\right],
    \]
    respectively, and it can be verified that  the norm of their product is $\| \differential\bbf |_{\bp^a} \differential\bbf |_{\bp^b}\|_{(1)}= 0.8750$.
    Interestingly, $\| \differential\bbf |_{\bp^b} \differential\bbf |_{\bp^a}\|_{(1)}= 0.7120$, which is different, but $<1$ too (as expected).  Stability can be ascertained by Proposition \ref{cor5}.
An expalantion, as pointed out by an anonymous referee, is that as a particular state gets ``more probable'', it actually is associated with ``less confidence", and hence there is indecision oscillating between alternatives.

\begin{center}
\begin{figure}[t]
  \centering
    \includegraphics[width=0.92\linewidth]{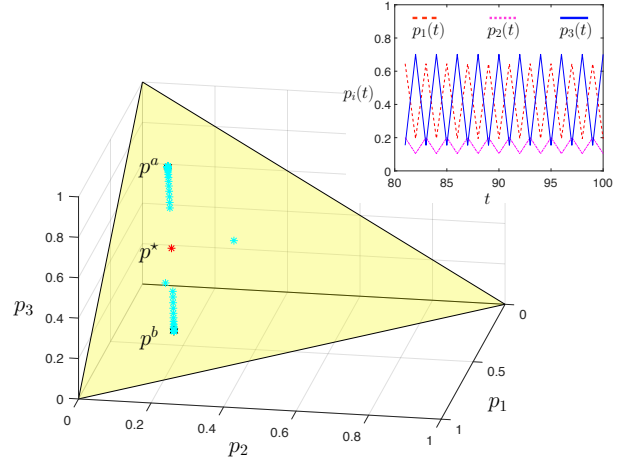}
    \vspace*{0.2in}
    \caption{For the $3$-state exponential model \eqref{repellingwalk}, $\gamma = 4$, and $\bC$ given by (\ref{periodicmatrixc2}), the unique fixed point  $\bp^{\star}=(0.4173,    0.1537,    0.4298)^\dagger$ is unstable and there is an attractive $2$-periodic orbit between $\bp^{a}$ and $\bp^{b}$, verified by the time history (inset graph).}
    \label{ExpRepellingExample13}
\end{figure} 
\end{center}

\begin{remark}\label{sec:continuous}



The framework presented extends naturally to continuous-time.
Indeed, a continuous-time analog of \eqref{eq:lazywalkboth} as a dynamical system on $\cS_{n-1}$ is given by
	\begin{equation}
	\dot\bp(t)= \bL^\dagger(\bI -\diag(\br(\bp(t))))\bp(t),
	\end{equation}
where $\bL=\bC-\bI$ is a Laplacian matrix satisfying $\bL\mathds{1}=0$. It is clear that $(\bI-\diag(\br(\bp(t))))\bL$ is a Laplacian matrix as $(\bI-\diag(\br(\bp(t))))\bL\mathds{1}=0$.  The scaling by $\diag(\br(\bp(t)))$ can be interpreted to play a similar role--it promotes or discourages staying at a state $i$ in accordance with the current value of the corresponding occupation probability $p_i$. 
A relation can be drawn by noting that for a small $h$,
	\begin{eqnarray*}
		\exp[h\bL^\dagger(\bI -\diag(\br(\bp)))] &\approx&\bI+h\bL^\dagger(\bI -\diag(\br(\bp)))
		\\&&\hspace*{-3cm}=(1-h)\bI +h(\,\underbrace{ \diag(\br(\bp))+\bC^\dagger(\bI -\diag(\br(\bp)))}_{\bPi(\br(\bp))}\,).
	\end{eqnarray*}
The special case when $\br(\bp)=\bp$ was recently considered in \cite{chen2017distributed}.
For general $\br(\cdot)$ as in \eqref{eq:options}, the existence of fixed points can be ascertained along similar lines as in the discrete-time setting.
%
\end{remark}

\section{DeGroot-Friedkin Model and its Variants}\label{sec:DGF}

We now consider the two classes of nonlinear Markov chains with $r(x)=\gamma x$ and $1 - \gamma x$, for $0<\gamma\leq 1$. The bounds $0<\gamma\leq 1$ ensure that $\bPi(\bp)$ (in \eqref{eq:lazywalk}) remains stochastic for all values of the probability vector $\bp$ and any $\bC$. For small values of $\gamma$, $\gamma x\simeq 1-e^{-\gamma x}$ and, evidently, these models approximate the corresponding exponential models of Section \ref{sec:Gs}.

\subsection{\bf Case $\mathbf{r(x)=\gamma x}$}\label{subsecLazyLinear}

The case where $r(x)=x$ and $\bC$ is restricted to be {\em doubly stochastic} has been studied in \cite{xu2015modified} and referred to as a {\em modified/one-step DeGroot-Friedkin} model. Existence and stability of the fixed point were analyzed and, in particular, it was conjectured that the equilibrium is stable for any irreducible row stochastic matrix $\bC$ (see \cite{xu2015modified}). Herein, we consider the general class where $r(x)=\gamma x$. For this class of models, very much as in the case of the exponential models, we can ascertain $\ell_{1}$ strict-contractivity for a range of values for $\gamma$, while for other values, we can ascertain stability on a case by case basis.
We begin with the following proposition for general irreducible stochastic $\bC$ and $\gamma\leq\frac{1}{2}$.\\

\begin{prop}\label{gammax}
	For $\gamma \leq \frac{1}{2}$ consider
\begin{subequations}\label{eq:lazywalklinear}
\begin{align}
	&\bp(t)\mapsto  \bbf(\bp(t))=\bp(t+1), \mbox{ where}\\
&\bp(t+1) = \left(\diag(\gamma\bp(t))+\bC^\dagger \diag(\mathds{1}-\gamma\bp(t))\right) \bp(t).\label{27b}
\end{align}
\end{subequations}
The map $\bbf$ is contractive in $\ell_1$, the iteration for $t=0,1,\ldots$ converges to a unique fixed point $\bp^\star=\lim_{t\to\infty}\bp(t)$, and
\begin{eqnarray}
\left(1-\gamma p^{\star}_{i}\right)p^{\star}_{i} = \kappa c_{i}, \mbox{ for a suitable } \kappa>0.
\label{FixedPointEqSeqDF}
\end{eqnarray}
\end{prop}
\begin{proof}
As before, the Jacobian $\differential\bbf$ is now
\begin{align*}
\bdelta\mapsto &\underbrace{\left( \diag(2\gamma\bp) + \bC^\dagger \diag(\mathds{1} - 2\gamma\bp) \right)}_{\bQ^{\dagger}(\bp)}\bdelta.
\end{align*}
For $0<\gamma\leq \frac{1}{2}$, $\bQ(\bp)$ is element-wise non-negative. Corollary \ref{cor1} ensures that $\bbf$ is contractive in $\ell_1$. The unique fixed point $\bp^{\star}$ satisfies
\[
\bC^{\dagger} (\mathds{1}-\gamma\bp^{\star})\odot\bp^{\star} = (\mathds{1}-\gamma\bp^{\star})\odot\bp^{\star},
\]
and therefore, $p_{i}^{\star}$ satisfies $\left(1-\gamma p^{\star}_{i}\right)p^{\star}_{i} = \kappa c_{i}$ with $\kappa = 1 - \gamma||\bp^{\star}||_{2}^{2}$. The global stability of $\bp^\star$ follows a similar argument as in Proposition \ref{thm:thm4}.
\end{proof}

For the range $\gamma \in[ \frac{1}{2},1]$ all-encompassing conclusions cannot be drawn and examples have to be worked out on a case by case basis. However, more can be said based on the induced norm of $\differential\bbf$ even when the elements of $\differential\bbf$ may have negative entries. Specifically, it is possible to obtain a closed-form expression for $\max_{\bp\in\cS_{n-1}}\|\differential\bbf|_\cT\|_{(1)}$ for $\gamma \in (\frac{1}{2},1]$. 
If $\bC$ has zero diagonal (a standard assumption in DeGroot-Friedkin literature), then for $\frac{1}{2} < \gamma < \frac{1}{2}\left(1 + \min_{i\neq j}C_{ji}\right)$, it can be shown that the map $\bbf$ remains $\ell_{1}$-contractive and consequently $\bp^\star$ in (\ref{FixedPointEqSeqDF}) remains globally attractive.
In passing, we note that for $\gamma=1$, trivially, the vertices of $\cS_{n-1}$ are fixed points while, in general, when $\gamma\neq 1$, this is not the case. Also, when $\bC$ is doubly stochastic and $\gamma\neq 1$,  $\frac{1}{n}\mathds{1}$ is the unique\footnote{That $\frac{1}{n}\mathds{1}$ is a fixed point can be verified directly, whereas the fact that there is no other fixed point can be argued in a similar manner as \cite[Theorem 2]{xu2015modified}.}
 fixed point of \eqref{eq:lazywalklinear}.


\subsection{\bf Case $\mathbf{r(x)=1-\gamma x}$}
We first establish that the corresponding map admits a unique fixed point for any $\gamma>0$, and show that it is $\ell_{1}$-contractive for $\gamma \leq \frac{1}{2}$.
\begin{prop}\label{oneminusgammax}
Consider 
\begin{subequations}\label{eq:repellingwalklinear}
\begin{align}
&\bp(t)\mapsto  \bbf(\bp(t))=\bp(t+1)\mbox{ where}\\
&\bp(t+1) = \left(\diag(\mathds{1}-\gamma\bp(t))+\bC^\dagger \diag(\gamma\bp(t))\right) \bp(t).
\end{align}
\end{subequations}
For any $\gamma>0$, there is a unique fixed point $\bp^\star$, where
\begin{eqnarray}
 p^{\star}_{i} = \displaystyle\frac{\sqrt{c}_{i}}{\displaystyle\sum_{i=1}^{n}\sqrt{c}_{i}}, \quad i=1,\hdots,n.
\label{RepellingLinearFixedPoint} 
\end{eqnarray}
For $0<\gamma\leq\frac{1}{2}$, $\bbf$ is $\ell_1$-contractive and in this case $\bp^\star$ is an attractive fixed point.
\end{prop}
\begin{proof}
The fixed-point condition 
\begin{align*}
\gamma \bp^{\star}\odot\bp^{\star} =  \bC^{\dagger}\gamma\bp^{\star}\odot\bp^{\star}
\end{align*}
implies that $p^{\star}_i$ must equal $\kappa \sqrt{c_i}$, for each $i$ and some $\kappa>0$. Thus, the fixed point is always unique and is as claimed.
For $0<\gamma\leq \frac{1}{2}$,
the Jacobian $\differential\bbf$ 
\begin{align*}
\bdelta\mapsto &\underbrace{\left( \diag(1-2\gamma\bp) + \bC^\dagger \diag(2\gamma\bp) \right)}_{\bQ^{\dagger}(\bp)}\bdelta
\end{align*}
is element-wise non-negative, inherits irreducibility from $\bC^{\dagger}$ in $\cS_{n-1}^o$, and as before, $\bbf$ is $\ell_1$-contractive.
\end{proof}

Once again, for $\gamma\in[\frac12,1]$, analysis can be done on a case by case basis and no general conclusion can be drawn. Similar to the comment in Section \ref{subsecLazyLinear}, we can find a closed-form expression for $\max_{\bp\in\cS_{n-1}}\|\differential\bbf|_\cT\|_{(1)}$ for $\gamma \in (\frac{1}{2},1]$. Then requiring $\max_{\bp\in\cS_{n-1}}\|\differential\bbf|_\cT\|_{(1)} < 1$, it can be shown that if $\bC$ has zero diagonal (a standard assumption in DeGroot-Friedkin literature), then for $\frac{1}{2} < \gamma < \frac{1}{2}\left(1 - \min_{i\neq j}C_{ij}\right)^{-1}$, the map $\bbf$ is guaranteed to be $\ell_{1}$-nonexpansive.

\section{Groupings}\label{sec:grouping}

It is quite interesting to speculate about the effect of colluding sub-group in opinion forming.
Indeed, everyday experience suggests that opinion is often reinforced within groups of like-minded individuals that draw confidence upon the collective wisdom, or lack of. To account for such interactions, we use a stochastic matrix $W$ to model the joint influence between group members by weighing their collective states via $\br(W\bp)$, which should be contrasted with individual-reinforcement of opinion/confidence modeled by $\br(\bp)$. This is independent and in addition to $\bC$, which is used to model information flow over the total influence network. A reasonable choice for $W$ is to be block diagonal where the blocks correspond to different subgroups of interacting individuals. The special case where $W$ is identity matrix reduces to the earlier setting.

In fact, what we propose herein is an ``interacting particle'' analogue for nonlinear Markov chains, modeled as follows:
\begin{eqnarray}\label{eq:lazywalkW}
&&\bp(t+1) = \bPi(\bp(t))^{\dagger}\bp(t) \nonumber\\
&&\hspace{-0.7cm}= \left(\diag(\br(W\bp(t)))+\bC^\dagger (\bI-\diag(\br(W\bp(t))))\right) \bp(t).
\end{eqnarray}
In particular, using a fixed-point argument as in \cite{jia2015opinion}, we establish existence results for the cases
$\br(x)=x$ and $\br(x)=1-e^{-x}$, and a general stochastic matrix $W$.

\begin{prop}
Let $\br(x)=x$ or $\br(x)=1-e^{-x}$, and $W$ a stochastic matrix. Assume that $\bc_k<\frac{1}{2}$ for all $k$. The Markov nonlinear model \eqref{eq:lazywalkW}, has at least one fixed point in the interior of probability simplex $\cS_{n-1}$.
\end{prop}
\begin{proof}
Any fixed point of \eqref{eq:lazywalkW} must satisfy 
	\[
		\bp_j = F_j(\bp):=\frac{1}{1+\frac{\sum_{k\neq j} \bc_k/(1-\br_k)}{\bc_j/(1-\br_j)}}.
	\]
Since 
	\[
		\sum_{k\neq j} \frac{\bc_k}{\bc_j(1-\br_k)} >\sum_{k\neq j} \frac{\bc_k}{\bc_j} >1,
	\]
there exists $\epsilon>0$ small enough such that
	\[
		\left(\sum_{k\neq j} \frac{\bc_k}{\bc_j(1-\br_k)}-1\right) \epsilon-\sum_{k\neq j}\frac{\bc_k}{\bc_j(1-\br_k)}
		\epsilon^2 >0.
	\]
It follows
	\[
		\frac{1}{1+\sum_{k\neq j} \frac{\bc_k}{\bc_j(1-\br_k)}} < 1-\epsilon.
	\]
Combining the above we obtain
	\[
		F_j(\bp) \le \frac{1}{1+\sum_{k\neq j} \frac{\bc_k}{\bc_j(1-\br_k)}}<1-\epsilon.
	\]
On the other hand, given 
	\[
		\bp \in {\cS_\epsilon}:=\{\bp \in \cS_{n-1} \mid \bp_i \le 1-\epsilon,~\forall i=1,\ldots,n\},
	\]
it is easy to see $\br(W\bp)\in \cS_\epsilon$ due to the facts that $\br(x)\le x$ and $W$ is stochastic. Thus, $F(\cS_\epsilon)\subset \cS_\epsilon$. Clearly, $F$ is continuous. Therefore, by Brouwer fixed-point theorem, there exists $\bp^\star$ such that $\bp^\star = F(\bp^\star)$.
\end{proof}

The ``nonlocal interaction'' matrix $W$ may in general introduce negative off-diagonal elements in $\differential\bbf$. The theory in Section \ref{sec:stability} applies on a case by case basis, but no general conclusion can be drawn at this point regarding global stability of particular class of models as we did earlier.
Indeed,  for $\br(x)=1-e^{-x}$, a matrix representation of the differential \eqref{eq:oneQ} becomes
		\begin{eqnarray*}
		\bQ(\bp)^\dagger &=& 
		\diag(\mathds{1} - e^{-W\bp}) + \bC^\dagger \diag(e^{-W\bp})\\
		&&+(\bI-\bC^\dagger) \diag(\bp\odot e^{-W\bp})W.
	\end{eqnarray*}
This, in general, has negative entries, which however doesn't imply that the fixed point is unstable. The theory in Section \ref{sec:stability} applies and attractiveness of equilibria can be ascertained by e.g., explicitly computing the $\ell_1$-gain of $\differential\bbf|_{\cT}$.
\begin{center}
\begin{figure}[t]
  \centering
    \includegraphics[width=0.87\linewidth]{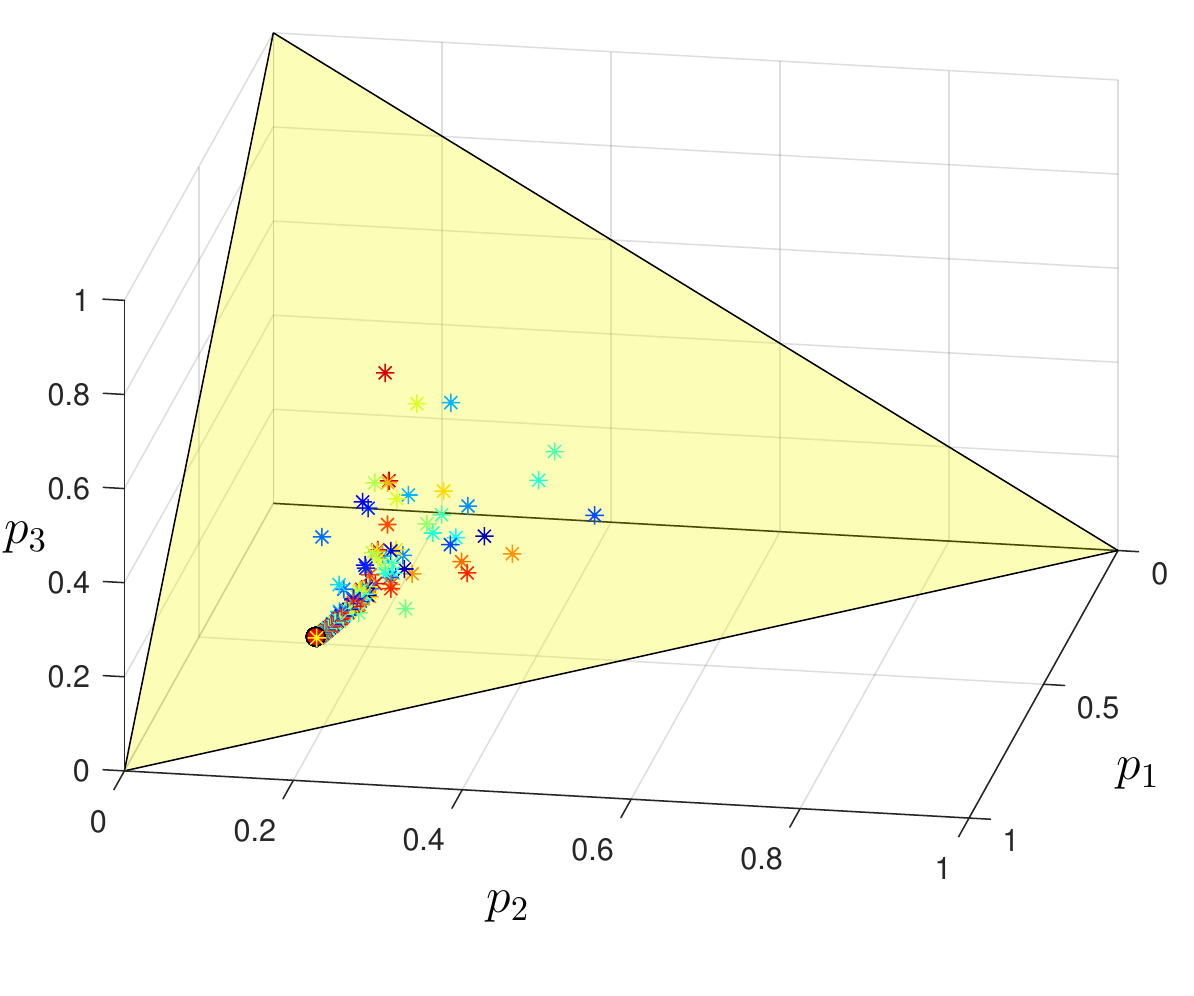}
    \caption{For the $3$-state model \eqref{eq:lazywalkW} with influence matrix $\bC$ and $W$ given by \eqref{subgroup_C_example} and \eqref{subgroup_W_example}, trajectories converge to the unique fixed point $\bp^{\star} = (0.6975, 0.1744, 0.1282)^{\dagger}$.}
    \label{subgroup_example}
\end{figure} 
\end{center}
Below is an example in $\cS_{2}$. We take $\br(x)=1-e^{-Wx}$,
	\begin{eqnarray}
	\bC=\left[\begin{matrix}0.8 & 0.1 & 0.1\\0.4 & 0.2 & 0.4\\ 0.4 & 0.4 & 0.2\end{matrix}\right]
    \label{subgroup_C_example}
	\end{eqnarray}
and
	\begin{eqnarray}
	W=\left[\begin{matrix}0.5 & 0.5 & 0\\0.5 & 0.5 & 0\\ 0 & 0 & 1\end{matrix}\right].
    \label{subgroup_W_example}
	\end{eqnarray}
    
Numerically (Fig. \ref{subgroup_example}), we see that the system has a unique fixed point, $\bp^{\star} = (0.6975, 0.1744, 0.1282)^\dagger$, which is stable.
This results are consistent with element-wise positiveness of the Jacobian of \eqref{eq:lazywalkW} which is evaluated at $\bp^{\star}$,
\[
 \differential\bbf |_{\bp^{\star}} =\left[\begin{matrix}0.8932 & 0.2812 & 0.3068\\0.0872 & 0.5052 & 0.3068\\ 0.0196 & 0.2136 & 0.3865\end{matrix}\right].
\]
It is worth mentioning that simulation with the same $\bC$ but this time with $W=I_{3\times 3}$ gives $\bp^{\star} = (0.8014, 0.0993, 0.0993)^\dagger$. Hence, as expected, the influence between member of the sub-group has a strengthening effect.   
\begin{remark}
The continuous space and time analogue of the nonlinear model \eqref{eq:lazywalkW} is the nonlinear Fokker-Plank-Vlasov equation
	\[
		\rho_t = \Delta \rho+\nabla\cdot(\rho\nabla V(x))+\nabla\cdot\left(\rho \nabla\left(\int W(x-y)\rho(y)\differential y\right)\right),
	\]
which is used to model the evolution of densities for interacting particles systems under the influence of external potential $V$ and interacting potential $W$ \cite{villani2004trend}. 
\end{remark}

\section{Concluding remarks and further research directions}\label{sec:conclusion}
\begin{center}
\begin{figure}[ht]
  \centering
    \includegraphics[width=0.68\linewidth]{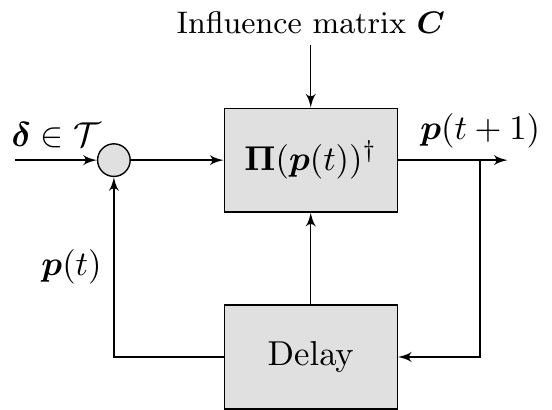}
    \vspace*{0.09in}
    \caption{The effect of bias.}
    \label{fig:feedback}
\end{figure} 
\end{center}
We presented conditions that guarantee global attractiveness of equilibria of nonlinear stochastic maps; these are Theorems \ref{thm1} and \ref{thm6} and Propositions \ref{cor3}, \ref{cor4}, and \ref{cor5} in Section \ref{sec:stability}. The criteria can be effectively used in certain cases where structural features can be exploited. Interest stems from modeling dynamical interactions over social networks. In Sections \ref{sec:Gs} and \ref{sec:DGF}, we highlight application of the theory in representative examples where the complementing statements of Section \ref{sec:stability} are pertinent, respectively. Section \ref{sec:grouping} presents a natural generalization of opinion models where the dynamics are modified by local interactions between subgroupings of the interacting agents. We expect that the development herein, i.e., both the theory as well as the new class of exponential models that we present in Sections \ref{sec:Gs} and \ref{sec:grouping} to provide impetus for further advances. 
In particular, a research direction of practical significance is to quantify the effect of uncertainty and disturbances in such models. For instance, Fig. \ref{fig:feedback} exemplifies the potential for bias/noise $\bm{\delta}$. Such a disturbance must belong to the tangent space of the probability simplex.
In general, it is of interest to quantify the effect of bias/disturbance in the dynamic response (e.g., shift in the position and nature of equilibria).

\section{Acknowledgments}

The research was supported in part by the NSF under Grant ECCS-1509387, the AFOSR under Grants FA9550-15-1-0045 and FA9550-17-1-0435, and the ARO under Grant W911NF-17-1-0429.

\spacingset{.9}
\bibliographystyle{IEEEtran}
\bibliography{refs}

\end{document}